\documentclass[conference]{IEEEtran}

\IEEEoverridecommandlockouts

\usepackage[T1]{fontenc}
\usepackage[utf8]{inputenc}
\usepackage{rotating}
\usepackage[final]{microtype}
\usepackage{color}
\usepackage{latexsym}
\usepackage{amssymb}
\usepackage{amsmath}
\usepackage{amsthm}
\usepackage{mathtools}
\usepackage{xspace}
\usepackage{multirow}
\usepackage{relsize}
\usepackage{wrapfig}
\usepackage{bussproofs}
\usepackage{graphicx}
\usepackage{mdframed}
\usepackage{cite}
\usepackage{hyperref}
\usepackage{turnstile}
\usepackage{centernot}
\usepackage[caption=false,font=footnotesize]{subfig}

\newcommand{\memstateset}{S}
\newcommand{\memstatesetarch}{\memstateset_\mathit{a}}
\newcommand{\memstatesetnarch}{\memstateset_{\overline{\mathit{a}}}}
\newcommand{\memstate}{s}
\newcommand{\memstatearch}{\memstate_\mathit{a}}
\newcommand{\memstatenarch}{\memstate_{\overline{\mathit{a}}}}

\newcommand{\memstatenarchargs}[1]{\memstate_{\overline{\mathit{a}},#1}}
\newcommand{\memstateargs}[1]{\memstate_{#1}}
\newcommand{\meminitstateset}{\memstateset_{I}}
\newcommand{\meminitstate}{\memstate_{I}}

\newcommand{\triggerstatesnoarg}{\memstateset_{b}}

\newcommand{\reachable}[1]{\mathit{reach}(#1)}
\newcommand{\bug}{\mathcal{B}}

\newcommand{\meminitstatenarch}{\memstatenarchargs{I}}

\newcommand{\memlocs}{\mathcal{L}}

\newcommand{\memloc}{l}
\newcommand{\memlocinv}{\memloc_{\mathit{orig}}}
\newcommand{\memlocinvy}{\memloc_{\mathit{dup}}}
\newcommand{\memlocorig}{\memloc_O}
\newcommand{\memlocdup}{\memloc_D}
\newcommand{\funcmemlocs}{\mathit{L}}
\newcommand{\funcmemlocsorig}{\memlocs_O}
\newcommand{\funcmemlocsdup}{\memlocs_D}
\newcommand{\funcmemlocscorr}{\funcmemlocs_\mathit{D}}
\newcommand{\funcmemlocscorrinv}{{\funcmemlocs_\mathit{D}}^{\!-1}}
\newcommand{\memvalue}{v}
\newcommand{\memvalues}{\mathcal{V}}
\newcommand{\funcmemlocsin}{\funcmemlocs_{\mathit{in}}}
\newcommand{\funcmemlocsout}{\funcmemlocs_{\mathit{out}}}
\newcommand{\funcdupsingle}{\mathit{Dup}}
\newcommand{\funcdupseq}{\funcdupsingle}
\newcommand{\qedcons}{\mathit{QEDcons}}
\newcommand{\instrset}{I}
\newcommand{\opcodeset}{\mathit{Op}}
\newcommand{\opcode}{\mathit{op}}
\newcommand{\funcopcode}{\mathit{op}}

\newcommand{\instrsetorig}{\instrset_O}
\newcommand{\instrsetdup}{\instrset_D}
\newcommand{\instr}{i}
\newcommand{\instrseq}{\boldsymbol{\instr}}
\newcommand{\instrargs}[1]{\instr_{#1}}
\newcommand{\instrorig}{\instr_O}
\newcommand{\instrorigseq}{\boldsymbol{\instrorig}}
\newcommand{\instrorigargs}[1]{\instr_{O,#1}}
\newcommand{\instrdup}{\instr_D}
\newcommand{\instrdupseq}{\boldsymbol{\instrdup}}
\newcommand{\instrdupargs}[1]{\instr_{D,#1}}

\newcommand{\transrel}{T}
\newcommand{\transrelmulti}{\transrel}
\newcommand{\statepath}{\boldsymbol{\memstate}}

\newcommand{\instrnop}{\instrargs{\mathit{nop}}}

\newcommand{\transsys}{\mathcal{P}}

\newcommand{\specfunc}{\mathit{Spec}}
\newcommand{\specfuncinstr}[2]{\specfunc_{#1}(#2)}
\newcommand{\specfuncinstrname}[1]{\specfunc_{#1}}

\newtheorem{definition}{Definition}
\newtheorem{proposition}{Proposition}
\newtheorem{corollary}{Corollary}
\newtheorem{lemma}{Lemma}
\newtheorem{theorem}{Theorem}
\newtheorem{example}{Example}

\begin{document}

\title{A Theoretical Framework for Symbolic \\ Quick Error Detection \thanks
  {This work was supported by the Defense Advanced Research
  Projects Agency, grant FA8650-18-2-7854. \textbf{Article to appear in Proc.~FMCAD 2020.}}}

\author{
  \IEEEauthorblockN{Florian Lonsing, Subhasish Mitra, and Clark Barrett}
\IEEEauthorblockA{
    Computer Science Department, Stanford University, Stanford, CA 94305, USA\\
    E-mail: \{lonsing, subh, barrett\}@stanford.edu
  }
}

\maketitle

\thispagestyle{plain}
\pagestyle{plain}

\begin{abstract}
Symbolic quick error detection (SQED) is a formal pre-silicon
verification technique targeted at processor designs. 
It leverages bounded model checking
(BMC) to check a design for counterexamples to a self-consistency
property: given the instruction set architecture (ISA) of the design,
executing an instruction sequence twice on the same inputs must always
produce the same outputs. Self-consistency is a universal,
implementation-independent property. Consequently, in contrast to traditional verification approaches that use
implementation-specific assertions (often generated manually), SQED does not require a
full formal design specification or manually-written properties.
Case studies have shown that SQED is
effective for commercial designs and that SQED substantially improves design
productivity. However, until now there has been no formal
characterization of its bug-finding capabilities.  We aim to
close this gap by laying a formal foundation for SQED. We use a
transition-system processor model and define the notion of a bug using an
abstract specification relation.
We prove the soundness of SQED, i.e., that any bug
reported by SQED is in fact a real bug in the processor.  Importantly, this result
holds regardless of what the actual specification relation is.
We next describe conditions under which SQED is complete, that is, what
kinds of bugs it is guaranteed to find.  We show that for a large class of bugs,
SQED can always find a trace exhibiting the bug. 
Ultimately, we prove full completeness of a variant of SQED that uses specialized state reset instructions.
Our results enable a rigorous understanding of SQED and its bug-finding
capabilities and give insights on how to optimize implementations of
SQED in practice.
\end{abstract}


%
\IEEEpeerreviewmaketitle


\section{Introduction}

\bstctlcite{IEEEexample:BSTcontrol}

Pre-silicon verification of HW designs given as models in a HW
description language (e.g., Verilog) is a critical step in
HW design.  Due to the steadily increasing
complexity of designs, it is crucial to detect logic design bugs
before fabrication to avoid more difficult and costly debugging in
post-silicon validation.

Formal techniques such as bounded model checking (BMC)~\cite{DBLP:conf/tacas/BiereCCZ99} have an advantage over
traditional pre-silicon verification techniques such as simulation in that they are exhaustive up to the
BMC bound. Hence, formal techniques provide valuable guarantees about the
correctness of a design under verification (DUV) with respect to the checked 
properties. However, in traditional assertion-based formal verification techniques,
these properties are implementation-specific and must be written manually based on
expert knowledge about the DUV.
Moreover, it is a well-known, long-standing challenge that sets of
manually-written, implementation-specific properties might be insufficient to
detect all bugs present in a DUV~\cite{DBLP:conf/charme/KatzGG99,DBLP:conf/tacas/ChocklerKV01,DBLP:conf/fmcad/Claessen07,DBLP:conf/date/GrosseKD07,DBLP:conf/dac/ChocklerKP10}.

Symbolic quick error detection (SQED)~\cite{DBLP:conf/itc/LinSBM15,8355908,DBLP:conf/date/SinghDSSGFSKBEM19,DBLP:conf/iccad/LonsingGMNSSYMB19} is a formal pre-silicon
verification technique targeted at processor designs. In sharp
contrast to traditional formal approaches, SQED does not require
manually-written properties or a formal specification of the
DUV. Instead, it checks whether a self-consistency~\cite{DBLP:conf/fmcad/JonesSD96}
property holds in the DUV.
The self-consistency property employed by SQED is universal and
implementation-independent. Each instruction in the instruction set architecture (ISA) of the DUV is
interpreted as a function in a mathematical sense.  The
self-consistency check then amounts to checking whether the outputs
produced by executing a particular instruction sequence match if the
sequence is executed twice, assuming the inputs to the two sequences also match.

SQED leverages BMC to exhaustively explore
all possible instruction sequences up to a certain length starting from a set of initial
states.  Several case studies have demonstrated that SQED is highly
effective at producing short bug traces by finding counterexamples to
self-consistency in a variety of processor designs, including industrial
designs~\cite{DBLP:conf/date/SinghDSSGFSKBEM19}. Moreover, SQED substantially increases verification productivity.

However, until now there has been no rigorous theoretical understanding of
(A) whether counterexamples to self-consistency found by
SQED always correspond to actual bugs in the DUV---the \emph{soundness} of
SQED---and (B) whether for each bug in the DUV there exists a
counterexample to self-consistency that SQED can find---the
\emph{completeness} of SQED.
This paper makes significant progress towards closing this gap.

We model a processor as a transition system.  This model abstracts
away implementation-level details, yet is sufficiently precise to
formalize the workings of SQED.
To prove soundness and (conditional) completeness of SQED, we need to
establish a correspondence between counterexamples to self-consistency
and bugs in a DUV. In our formal model we achieve this correspondence
by first defining the correctness of instruction executions by means of a
general, abstract specification.  A bug is then a violation of this
specification. The abstract specification expresses the following general and
natural property we expect to hold for actual DUVs: an instruction
writes a correct output value into a destination location and does not modify
any other locations.

As \textbf{our main results}, we prove soundness and conditional completeness of
SQED. For soundness, we
prove that if SQED reports a counterexample to the universal
self-consistency property, then the processor has a bug. This result
shows that SQED does not produce spurious counterexamples.
Importantly, this result holds regardless of the actual specification,
confirming that SQED does not depend on such implementation-specific details.
For completeness, we prove that if the processor has a bug then, under modest
assumptions, there exists a counterexample to self-consistency that can
be found by SQED. 
We also show that SQED can be made fully (unconditionally) complete with
additional HW support in the form of specialized state reset instructions.
Our results enable a rigorous
understanding of SQED and its bug-finding capabilities in actual
DUVs and provide
insight on how to optimize implementations of SQED.

In the following, we first present an overview of SQED from a
theoretical perspective (Section~\ref{sec:sqed:overview}). 
Then we define our
transition system model of processors (Section~\ref{sec:formal:model})
and formalize the correctness of instruction executions in terms of an
abstract specification relation (Section~\ref{sec:abs:spec}). After
establishing a correspondence between the abstract specification and
the self-consistency property employed by SQED
(Section~\ref{sec:qed:cons}), we prove 
soundness and (conditional) completeness of SQED
(Section~\ref{sec:main:proofs}). We conclude with a discussion of
related work and future research directions
(Sections~\ref{sec:related:work} and ~\ref{sec:conclusion}).


\section{Overview of SQED} \label{sec:sqed:overview}

We first informally introduce the basic concepts and terminology related to
SQED. 
Fig.~\ref{fig:sqed:overview} shows an overview of the high-level workflow.
\begin{figure*}
  \centering
  \subfloat[]{\includegraphics[scale=0.45]{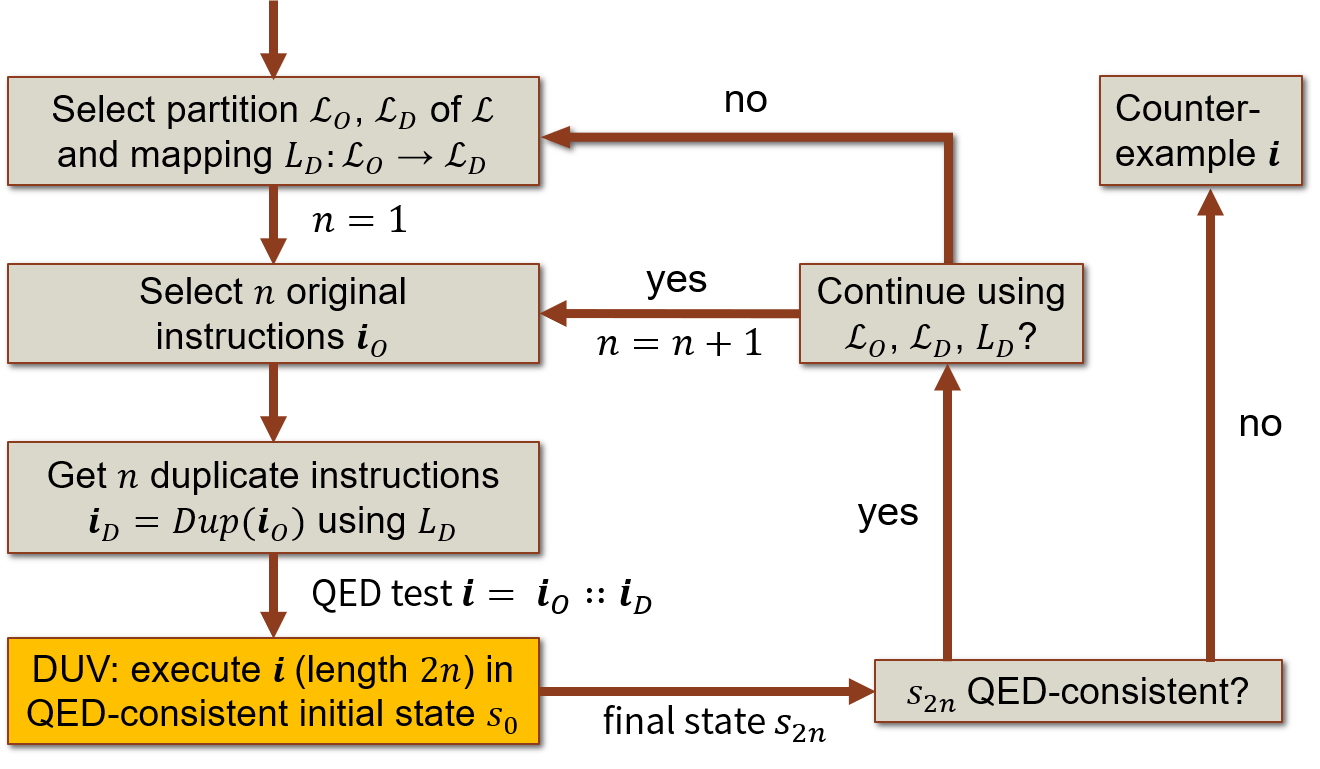} \label{fig:sqed:overview}}
  \hspace*{0.125cm}
  \subfloat[]{\raisebox{1.5cm}{\includegraphics[scale=0.45]{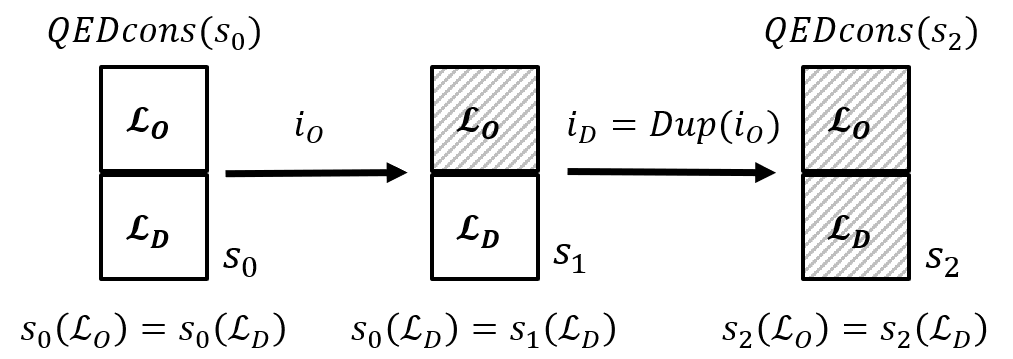}} \label{fig:sqed:qedcons}}
\caption{SQED workflow from a theoretical perspective (a) and
  illustration of executing the QED test $\instrseq = \instrorig ::
  \instrdup$ in Example~\ref{ex:qed:test} (b).}
\label{fig:sqed}
\end{figure*}
Given a processor design $\transsys$, i.e., the DUV, SQED is based
on symbolic execution of instruction sequences using BMC. We assume
that an \emph{instruction} $\instr = (\opcode, \memloc, (\memloc',\memloc''))$ consists of an opcode $\opcode$, an output
location $\memloc$, and a pair $(\memloc',\memloc'')$ of input
locations.\footnote{This model is used for simplicity, but it could
easily be extended to allow instructions with additional inputs or outputs.}
\emph{Locations} are an abstraction used to represent registers and memory locations.

The self-consistency check is based on executing two instructions that should always produce the same result.
The two instructions are called an 
\emph{original} and a \emph{duplicate instruction}, respectively.
The duplicate instruction has the \emph{same opcode} as the
original one, i.e., it implements the same functionality, but it
operates on different input and output locations. The locations
on which the duplicate instruction operates are determined by an \emph{arbitrary but fixed} \emph{bijective function} $\funcmemlocscorr:
\funcmemlocsorig \rightarrow \funcmemlocsdup$ between two subsets
$\funcmemlocsorig$, the \emph{original locations}, and
$\funcmemlocsdup$, the \emph{duplicate locations}, that form a
partition of the set $\memlocs$ of all locations in $\transsys$.
An original instruction can only use locations in $\funcmemlocsorig$.  An \emph{instruction
 duplication function} $\funcdupsingle$ then maps any original
instruction $\instrorig$ to its duplicate $\instrdup$ by copying the opcode and then
applying $\funcmemlocscorr$ to its locations.

\begin{example} \label{ex:sqed:instr:dup}
Let $\memlocs = \{0, \ldots, 31\}$ be the identifiers of~32
registers of a processor $\transsys$, and consider the partition
$\funcmemlocsorig = \{0,1, \ldots, 15\}$ and $\funcmemlocsdup =
\{16,17, \ldots, 31\}$.  Let $\instrorig = (\mathsf{ADD}, \memloc_{12},
(\memloc_4,\memloc_8))$ be an original register-type ADD instruction
operating on registers $4, 8$, and $12$.
Using $\funcmemlocscorr(k) = k + 16$, we obtain 
$\funcdupsingle(\instrorig) = \instrdup =
(\mathsf{ADD}, \memloc_{28}, (\memloc_{20},\memloc_{24}))$. 

Consider a different partition $\funcmemlocsorig' = \{0,2,4, \ldots,
30\}$ and $\funcmemlocsdup' = \{1,3,5, \ldots, 31\}$ and 
function $\funcmemlocscorr'(k) = k+1$. For this function,
$\funcdupsingle(\instrorig) = (\mathsf{ADD}, \memloc_{13},
(\memloc_{5},\memloc_{9}))$.
\end{example}

Self-consistency checking is implemented using \emph{QED
tests}. A QED test is an instruction sequence $\instrseq =
\instrorigseq :: \instrdupseq$ consisting of a sequence
$\instrorigseq$ of $n$ original instructions followed by a corresponding
sequence $\instrdupseq = \funcdupsingle(\instrorigseq)$ of $n$
duplicate instructions (where operator ``::'' denotes concatenation).
A QED test $\instrseq$ is symbolically executed from a \emph{QED-consistent state}, that is,
a state where the value stored in each original location $\memloc$ is the same
as the value stored in its corresponding duplicate location $\funcmemlocsdup(\memloc)$.
The resulting final state after executing $\instrseq$ should then also be QED-consistent.
Fig.~\ref{fig:sqed:overview} illustrates the
workflow. A QED test $\instrseq$ \emph{succeeds} if the final state
that results from executing $\instrseq$ is QED-consistent; otherwise
it \emph{fails}.  Starting the execution in a QED-consistent state
guarantees that original and duplicate instructions receive the same
input values. Thus, if the final state is not QED-consistent, then this
indicates that some pair of original and duplicate instructions behaved differently.

\begin{example} \label{ex:qed:test}
Consider Fig.~\ref{fig:sqed:qedcons} and the QED test $\instrseq = \instrorig :: \instrdup$ consisting
of one original instruction $\instrorig$ and its duplicate
$\funcdupsingle(\instrorig) = \instrdup$ for some function $\funcmemlocscorr$. Suppose that 
$\instrseq$ is executed in a QED-consistent state $\memstateargs{0}$ 
(denoted by $\qedcons(\memstateargs{0})$ and $\memstateargs{0}(\funcmemlocsorig) =
\memstateargs{0}(\funcmemlocsdup)$) and both
$\instrorig$ and $\instrdup$ execute correctly. Instruction
$\instrorig$ produces state $\memstateargs{1}$, where the values
at duplicate locations remain unchanged, i.e.,
$\memstateargs{0}(\funcmemlocsdup) = \memstateargs{1}(\funcmemlocsdup)$, because $\instrorig$ operates on original
locations only.  When instruction $\instrdup$ is executed in state
$\memstateargs{1}$, it modifies only duplicate locations. The final state
$\memstateargs{2}$ is QED-consistent (denoted by
$\qedcons(\memstateargs{2})$ and $\memstateargs{2}(\funcmemlocsorig) =
\memstateargs{2}(\funcmemlocsdup)$), and thus QED
test $\instrseq$ succeeds.
\end{example}

\begin{example}[Bug Detection] \label{ex:bug:detection}
Consider processor $\transsys$ and $\funcmemlocsorig$
and $\funcmemlocsdup$ from
Example~\ref{ex:sqed:instr:dup}. Let $\instrorigargs{1} =
(\mathsf{ADD}, \memloc_{12}, (\memloc_4,\memloc_{15}))$ and
$\instrorigargs{2} = (\mathsf{MUL}, \memloc_{15},
(\memloc_{12},\memloc_{12}))$ be original register-type addition
and multiplication instructions. Using $\funcmemlocscorr(k) = k + 16$,
we obtain $\funcdupsingle(\instrorigargs{1}) = \instrdupargs{1} =
(\mathsf{ADD}, \memloc_{28}, (\memloc_{20},\memloc_{31}))$ and
$\funcdupsingle(\instrorigargs{2}) = \instrdupargs{2} = (\mathsf{MUL},
\memloc_{31}, (\memloc_{28},\memloc_{28}))$.  Assume that $\transsys$
has a bug that is triggered when two MUL instructions are executed in
subsequent clock cycles, resulting in the corruption of the
output location of the second MUL instruction.\footnote{This
  scenario corresponds to a real bug in an out-of-order RISC-V
  design detected by SQED:
  \url{https://github.com/ridecore/ridecore/issues/4}.}  
Note that executing the QED test $\instrseq =
\instrorigargs{1},\instrorigargs{2} ::
\instrdupargs{1},\instrdupargs{2}$ in a QED-consistent initial state
produces a QED-consistent final state: the bug is not triggered by
$\instrseq$ because $\instrdupargs{1}$ is executed between $\instrorigargs{2}$
and $\instrdupargs{2}$.  A slightly longer test $\instrseq =
\instrorigargs{2},\instrorigargs{1},\instrorigargs{2} ::
\instrdupargs{2},\instrdupargs{1},\instrdupargs{2}$ does trigger the bug,
however, because the subsequence $\instrorigargs{2},\instrdupargs{2}$ of two back-to-back MULs
causes the first duplicate instruction $\instrdupargs{2}$ in $\instrseq$ to produce an incorrect
result at $\memloc_{31}$.  This incorrect result then propagates through the next two instructions,
resulting in a QED-inconsistent final state since the values at
$\memloc_{15}$ and $\memloc_{31}$, i.e., the output locations of
$\instrorigargs{2}$ and $\instrdupargs{2}$, differ.
\end{example}

QED-consistency is the universal, implementation-indepen\-dent property that is
checked in SQED. In practice, the property must refer to some basic
information about the design such as, e.g., symbolic register names, but this can
be generated automatically from a high-level ISA description~\cite{DBLP:conf/iccad/LonsingGMNSSYMB19}.
BMC is used to symbolically and exhaustively generate
all possible QED tests up to a certain length $2n$ (the BMC
bound).  BMC ensures that SQED will find the shortest
possible failing QED test first. The high-level workflow shown in
Fig.~\ref{fig:sqed:overview} allows for flexibility in choosing
the partition and mapping between original and duplicate
locations.  We rely on this flexibility for the results in this paper (Theorems~\ref{thm:main}
and~\ref{thm:hard:reset}).
Current SQED implementations 
use a predefined partition and mapping, based on which BMC
enumerates all possible QED tests. 
Extending implementations to have the BMC tool also choose a
partition and mapping could be explored in future work.

We refer to related
work~\cite{DBLP:conf/itc/LinSBM15,DBLP:conf/date/FadihehUNMBSK18,DBLP:conf/date/SinghDSSGFSKBEM19} 
for case studies that demonstrate the effectiveness of BMC-based SQED on a
variety of processor designs.
The scalability of SQED in
practice is determined by the scalability of the BMC tool being used. Thus, approaches
for improving scalability of BMC can also be applied to SQED,
e.g. abstraction, decomposition, and partial instantiation
techniques~\cite{DBLP:conf/itc/LinSBM15}.


\section{Instruction and Processor Model} \label{sec:formal:model}

We model a processor as a transition system containing an abstract set
of locations. The set of locations includes registers and memory
locations. A state of a processor consists of an \emph{architectural}
and a \emph{non-architectural} part. In a state transition that
results from executing an instruction, the architectural part of a
state is modified explicitly by updating the value at the output
location of the executed instruction. The architectural
part of a state is also called the \emph{software-visible} state of
the processor. It comprises those parts of the state that can be
updated by executing instructions of the user-level ISA of the
processor, such as memory locations and general-purpose registers. The
non-architectural part of a state comprises the remaining parts that
are updated only implicitly by executing an instruction, such as
pipeline or status registers.

Instructions are functions that take inputs from locations and write
an output to a location. We assume that every instruction produces its
result in one transition.  In our model, we abstract away
implementation details of complex processor designs (e.g.,
pipelined, out-of-order, multi-processor systems). This is for ease of
presentation and reasoning. However, many of these complexities can be
viewed as refinements of our abstraction, meaning that our formal results
still hold on complex models (i.e., our results
can be lowered to more detailed models such as those described
in~\cite{DBLP:conf/itc/LinSBM15,8355908}).  Working out the details of such
refinements is one important avenue for future work. 

\begin{definition}[Transition System]\label{def:revised:transsys}
  A processor is a \emph{transition system}~\cite{DBLP:conf/sagamore/Keller74,DBLP:journals/cacm/Keller76} $\transsys = (\memvalues, \memlocs, \memstatesetnarch, \meminitstatenarch, \opcodeset, \instrset, \transrel)$, where
  \begin{itemize}
  \item $\memvalues$ is a set of \emph{abstract data values}, 
  \item $\memlocs$ is a set of \emph{memory locations} (from which we
    define the set $\memstatesetarch$ of \emph{architectural states} as the set
    of total functions from locations to values, i.e. 
    $\memstatesetarch = \{\memstatearch \mid \memstatearch: \memlocs
    \rightarrow \memvalues\}$),
  \item $\memstatesetnarch$ is a set of \emph{non-architectural states} (from
    which we further define the set of all \emph{states} as $S = \memstatesetarch \times \memstatesetnarch$),
  \item $\meminitstatenarch\in\memstatesetnarch$ is a unique \emph{initial non-architectural state} (from which
    we define the set of \emph{initial states} as $\meminitstateset =
    \memstatesetarch \times \{ \meminitstatenarch \}$,
  \item $\opcodeset$ is a set of \emph{operation codes (opcodes)}, 
  \item $\instrset = \opcodeset \times \memlocs \times \memlocs^2$ is the set of \emph{instructions}, and 
  \item $\transrel: \memstateset \times \instrset \rightarrow \memstateset$ is the \emph{transition function}, which is total.
  \end{itemize}
\end{definition}

A state $\memstate \in \memstateset$ with $\memstate =
(\memstatearch,\memstatenarch)$ consists of an architectural part
$\memstatearch \in \memstatesetarch$ and a non-architectural  
part $\memstatenarch \in \memstatesetnarch$.
In the architectural part
$\memstatearch: \memlocs \rightarrow \memvalues$, $\memlocs$
represents all possible registers and memory locations, i.e., in
practical terms, $\memlocs$ is the address space of $\transsys$.
An initial state $\meminitstate \in \meminitstateset$ with
$\meminitstate = (\memstatearch, \meminitstatenarch)$ is defined by a
unique non-architectural part $\meminitstatenarch \in
\memstatesetnarch$ and an arbitrary architectural part $\memstatearch
\in \memstatesetarch$. 
We assume that $\meminitstatenarch \in \memstatesetnarch$ is unique to
make the exposition simpler. Our model could easily be extended to a
set of initial non-architectural states. 
The number $|\memlocs|$ of memory locations is arbitrary but fixed. 
We write $\memvalue =
\memstate(\memloc)$ to denote the value $\memvalue =
\memstatearch(\memloc)$ at location $\memloc \in \memlocs$
in state $\memstate = (\memstatearch,\memstatenarch)$. 
We also write
$(\memvalue, \memvalue') = \memstate(\memloc,\memloc')$ as shorthand
for $\memvalue = \memstate(\memloc)$ and \nolinebreak $\memvalue' =
\memstate(\memloc')$.

To formally define instruction duplication, we need to reason about \emph{original}
and \emph{duplicate} memory locations. To this end, we partition the set $\memlocs$
of memory locations into two sets of equal size, the \emph{original} and \emph{duplicate
locations} $\funcmemlocsorig$ and $\funcmemlocsdup$, respectively,
i.e., $\funcmemlocsorig \cap \funcmemlocsdup = \emptyset$, 
$\funcmemlocsorig \cup \funcmemlocsdup = \memlocs$, and $|\funcmemlocsorig| =
|\funcmemlocsdup|$.  Given
$\funcmemlocsorig$ and $\funcmemlocsdup$, we define an \textbf{arbitrary but
fixed} \emph{bijective function} $\funcmemlocscorr: \funcmemlocsorig
\rightarrow \funcmemlocsdup$ that maps an original location
$\memlocorig \in \funcmemlocsorig$ to its corresponding duplicate
location $\memlocdup =
\funcmemlocscorr(\memlocorig)$. The inverse of $\funcmemlocscorr$ is
denoted by $\funcmemlocscorrinv$ and is uniquely
defined.
We write $(\memlocdup,
\memlocdup') = \funcmemlocscorr(\memlocorig, \memlocorig')$ as shorthand for
$\memlocdup = \funcmemlocscorr(\memlocorig)$ and $\memlocdup' =
\funcmemlocscorr(\memlocorig')$.
Function $\funcmemlocscorr$ implements a
correspondence between original and duplicate locations, which we need 
to define QED-consistency (Definition~\ref{def:qed:consistency} below).

An instruction $\instr \in \instrset$ with $\instr = (\opcode,
\memloc, (\memloc',\memloc''))$ is defined by an opcode
$\opcode \in \opcodeset$, an output location $\memloc \in \memlocs$, and a pair of input locations $(\memloc',\memloc'') \in 
\memlocs^2$.
Function $\funcopcode: \instrset \rightarrow \opcodeset$ maps an
instruction to its opcode $\funcopcode(\instr)$. 
Functions $\funcmemlocsout: \instrset \rightarrow
\memlocs$ and $\funcmemlocsin: \instrset \rightarrow
\memlocs^2$ map an instruction $\instr$ to its output and input
locations $\funcmemlocsout(\instr) = \memloc$ and
$\funcmemlocsin(\instr) = (\memloc',\memloc'')$, respectively.
Given a state $\memstate = (\memstatearch,\memstatenarch)$, instruction $\instr$ reads values in
$\memstate$ from its input locations $\funcmemlocsin(\instr)$ and
writes a value to its output location $\funcmemlocsout(\instr)$,
resulting in a transition to a new state $\memstate' = (\memstatearch',\memstatenarch')$, written as
$\memstate' = \transrel(\memstate, \instr)$. The
transition function $\transrel$ is total, i.e., for every instruction
$\instr$ and state $\memstate$, there exists a successor state
$\memstate' = \transrel(\memstate, \instr)$.
As mentioned above, we have kept the model simple in order to make the
presentation more accessible, but our results can be lifted to many extensions,
including, e.g., more complicated kinds of instructions or instructions with
enabledness conditions \nolinebreak cf.~\cite{DBLP:journals/todaes/HuangZSVGM19}.

We write $\instrseq \in \instrset^{n}$ and $\statepath \in
\memstateset^{n}$ to denote sequences $\instrseq = \langle
\instrargs{1}, \ldots, \instrargs{n} \rangle$ and $\statepath =
\langle \memstateargs{1}, \ldots, \memstateargs{n} \rangle$ of $n$
instructions and $n$ states, respectively. We will use $::$ for sequence
concatenation and extend the transition
function $\transrel$ to sequences as follows.

\begin{definition}[Path] \label{def:paths}
Given sequences $\instrseq = \langle \instrargs{1}, \ldots,
\instrargs{n} \rangle$ and $\statepath = \langle \memstateargs{1},
\ldots, \memstateargs{n} \rangle$ of $n$ instructions and states,
$\statepath$ is a \emph{path} from state $\memstate_0 \in \memstateset$ to $\memstate_n$ via
$\instrseq$, written $\statepath = \transrelmulti(\memstateargs{0},
\instrseq)$, iff $\bigwedge^{n-1}_{k =
  0}  \memstateargs{k+1} = \transrel(\memstateargs{k}, \instrargs{k+1})$.
\end{definition}

\noindent
If $\statepath = \transrelmulti(\memstateargs{0},
\instrseq)$, then
for convenience we also write $\memstateargs{n} =
\transrelmulti(\memstateargs{0}, \instrseq)$ to denote the final state
$\memstateargs{n}$.

\begin{definition}[Reachable State] \label{def:reachable:state}
A state $\memstate$ is \emph{reachable}, written $\reachable{\memstate}$,
iff $\memstate = \transrelmulti(\memstateargs{0}, \instrseq)$ for some
$\memstateargs{0} \in \meminitstateset$ and instruction sequence $\instrseq$.
\end{definition}

The set $\instrset$ of instructions contains as proper subsets
the sets of \emph{original} and \emph{duplicate instructions}, $\instrsetorig$
and $\instrsetdup$, respectively.
Original (duplicate) instructions operate only on original (duplicate)
locations, i.e., $\forall \instrorig \in
\instrsetorig.\ \funcmemlocsin(\instrorig) \in \funcmemlocsorig^2
\wedge \funcmemlocsout(\instrorig) \in \funcmemlocsorig$ and
$\forall \instrdup \in \instrsetdup.\ \funcmemlocsin(\instrdup)
\in \funcmemlocsdup^2 \wedge \funcmemlocsout(\instrdup) \in
\funcmemlocsdup$. Given these definitions, we formalize instruction
duplication as follows.

\begin{definition}[Instruction Duplication] \label{def:single:instr:dup}
Let $\funcdupsingle: \instrsetorig \rightarrow \instrsetdup$ be an \emph{instruction
duplication function} that maps an original instruction $\instrorig = (\opcode,
\memlocorig, (\memlocorig',\memlocorig''))$ to a duplicate instruction
$\instrdup = \funcdupsingle(\instrorig) = (\opcode,
\funcmemlocscorr(\memlocorig), \funcmemlocscorr(\memlocorig', \memlocorig''))$
with respect to the bijective function $\funcmemlocscorr$.
\end{definition}

\noindent
An original instruction and its duplicate have the same opcode. We write $\instrorigseq \in \instrsetorig^{n}$ and $\instrdupseq \in
\instrsetdup^{n}$ to denote sequences $\instrorigseq = \langle
\instrorigargs{1}, \dots, \instrorigargs{n} \rangle$ and $\instrdupseq
= \langle \instrdupargs{1}, \dots, \instrdupargs{n} \rangle$ of $n$
original and duplicate instructions, respectively.  We lift
$\funcdupsingle$ in the natural way also to sequences of instructions
as follows.

\begin{definition}[Instruction Sequence Duplication]\label{def:seq:instr:dup}
Let $\instrorigseq = \langle \instrorigargs{1}, \dots, \instrorigargs{n}
\rangle$ be a sequence of original instructions.  Then 
$\funcdupseq(\instrorigseq) = \langle \funcdupsingle(\instrorigargs{1}), \dots,
\funcdupsingle(\instrorigargs{n}) \rangle$.
\end{definition}


\section{Formalizing Correctness} \label{sec:abs:spec}

We formalize the correctness of instruction executions in a processor
$\transsys$ using an abstract specification relation. We then
link this abstract specification to QED-consistency, the
self-consistency property employed by SQED (Section~\ref{sec:qed:cons}
below).

For our formalization, we assume that every opcode $\opcode \in \opcodeset$
has a \emph{specification function} $\specfuncinstrname{\opcode}: \memvalues^2
\rightarrow \memvalues$ that specifies how the opcode computes an output value
from input values.  Using this family of functions, we define an overall
\emph{abstract specification relation} $\specfunc \subseteq \memstateset \times
\instrset \times \memstateset$, which expresses when an instruction $\instr \in \instrset$ can transition to a state
$\memstate' \in \memstateset$ from a state $\memstate \in
\memstateset$ while respecting the opcode specification. 

\begin{definition}[Abstract Specification] \label{def:spec} $\forall\, \memstate, \memstate' \in \memstateset,\,  \instr \in \instrset.$
  \begin{align}
     \specfunc( & \memstate, \instr, \memstate') \leftrightarrow  \forall \memloc \in \memlocs.\ \nonumber \\  
     & (\memloc \not = \funcmemlocsout(\instr) \rightarrow \memstate(\memloc) = \memstate'(\memloc)) \wedge{} \label{assume:trans:spec} \\ 
   &  (\memloc = \funcmemlocsout(\instr) \rightarrow \memstate'(\memloc) = \specfuncinstr{\funcopcode(\instr)}{\memstate(\funcmemlocsin(\instr))})  \nonumber  
  \end{align}
\end{definition}

\noindent
Equation~\eqref{assume:trans:spec} states general and
natural properties that we expect to hold for a processor $\transsys$. If an
instruction $\instr$ executes according to its specification, then the
values at locations that are not output locations of 
$\instr$ are unchanged.
Additionally, the value produced at the output location of the instruction must agree with the value specified by
function $\specfuncinstrname{\funcopcode(\instr)}$.
Note that the specification relation $\specfunc$ specifies only how
the architectural part of a state is updated by a transition (not
the non-architectural part). Consequently, there might exist multiple
states whose non-architectural parts satisfy the right-hand side of~\eqref{assume:trans:spec}. This is why
$\specfunc$ is a relation rather than a function.
As special cases of~\eqref{assume:trans:spec}, original and duplicate instructions have the following properties:
  \begin{align}
\forall \memstate, \memstate' \in \memstateset,\  & \instrorig \in 
\instrsetorig, \memlocorig \in \funcmemlocsorig, \instrdup \in
\instrsetdup, \memlocdup \in \funcmemlocsdup. \nonumber \\
  & ( \specfunc(\memstate, \instrorig, \memstate') \rightarrow
 \memstate(\memlocdup) = \memstate'(\memlocdup) ) \wedge{} \label{cor:trans:spec:orig:instr} \\
  & ( \specfunc(\memstate, \instrdup, \memstate') \rightarrow 
\memstate(\memlocorig) = \memstate'(\memlocorig) ) \label{cor:trans:spec:dup:instr} 
    \end{align}

\noindent
Equations~\eqref{cor:trans:spec:orig:instr}
and~\eqref{cor:trans:spec:dup:instr} express that the
execution of an original (duplicate) instruction does not change the
values at duplicate (original) locations if the
instruction executes according to its specification.
The following \emph{functional congruence} property of instructions also follows from~\eqref{assume:trans:spec}:
\begin{align}
  \forall\, & \memstate_0,  \memstate_1, \memstate', \memstate'' \in
  \memstateset, \instr,\instr' \in \instrset. \nonumber \\ 
  \big[ & \funcopcode(\instr)=\funcopcode(\instr') \wedge \specfunc( \memstate_0, \instr, \memstate') \wedge   
     \specfunc(\memstate_1, \instr', \memstate'') \wedge{} \label{def:func:congr:instr} \\  
     & \memstate_0(\funcmemlocsin(\instr))
    = \memstate_1(\funcmemlocsin(\instr'))
    \big] \rightarrow \memstate'(\funcmemlocsout(\instr)) = \memstate''(\funcmemlocsout(\instr')) \nonumber
\end{align}

\noindent
By functional congruence, if two instructions with the same opcode are executed on inputs with the same
values, then the output values are the same.
We next define the correctness of a
processor $\transsys$ based on the abstract specification $\specfunc$.

\begin{definition}[Correctness] \label{def:correctness:weak:new:bug}
  A processor $\transsys$ is \emph{correct} with respect to
  specification $\specfunc$ iff  
  $\forall\, \instr \in \instrset, \memstate \in
    \memstateset.\: \reachable{\memstate}  \rightarrow 
  \specfunc(\memstate, \instr, \transrelmulti(\memstate, \instr))$.
\end{definition}
\noindent
Correctness requires every instruction to execute according to the abstract specification $\specfunc$ in 
every reachable state of $\transsys$.

A \emph{bug} in $\transsys$ is a counterexample to correctness, i.e., an
instruction that fails in at least one (not necessarily initial) reachable
state and may or may not fail in other
states. 

\begin{definition}[Bug] \label{def:bug:pair}
A \emph{bug with respect to specification $\specfunc$} in a processor $\transsys$ is defined by a
pair $\bug = \langle \instr_b, \triggerstatesnoarg \rangle$ consisting of
an instruction $\instr_b \in \instrset$ and a non-empty set $\triggerstatesnoarg
\subseteq \memstateset$ of states such that $\triggerstatesnoarg =
\{\memstate \in \memstateset \mid \reachable{\memstate} \wedge \neg \specfunc(\memstate, \instr_b, \transrelmulti(\memstate,
\instr_b))\}$.
\end{definition}

\noindent
The above definitions rely on the notion of an abstract specification
relation. Having \emph{some} abstract specification is a
\emph{theoretical} construct that is necessary to formally
characterize instruction failure and establish formal proofs about
SQED. However, it is important to note that to apply SQED in \emph{practice},
we do not need to know what the abstract specification relation is.

A bug $\langle \instr_b, \triggerstatesnoarg \rangle$ is precisely
characterized by the set $\triggerstatesnoarg$ of all reachable states in which
$\instr_b$ fails. The following proposition follows
from Definitions~\ref{def:correctness:weak:new:bug}
and~\ref{def:bug:pair}.

\begin{proposition} \label{prop:bug:correctness:new}
A processor $\transsys$ has a bug with respect to
  specification $\specfunc$ iff it is not correct with respect to $\specfunc$.
\end{proposition}

As special cases of processor correctness and bugs, respectively, we define correctness and
bugs with respect to instructions that are executed in an initial state only.

\begin{definition}[Single-Instruction Correctness] \label{def:si:correctness}
  Processor $\transsys$ is \emph{single-instruction correct} iff:
\[\forall\, \instr \in \instrset,
  \memstateargs{0} \in \meminitstateset.\:
  \specfunc(\memstateargs{0}, \instr, \transrelmulti(\memstateargs{0}, \instr)).\]
\end{definition}

\noindent
Single-instruction correctness implies that all instructions, i.e.,
all opcodes and all combinations of input and output locations,
execute correctly in all initial states. A \emph{single-instruction 
bug} is a counterexample to single-instruction correctness.

\begin{definition}[Single-Instruction Bug] \label{def:si:bug}
Processor $\transsys$ has a \emph{single-instruction bug} with respect to specification $\specfunc$ iff
$\exists\, \instr \in \instrset,
  \memstateargs{0} \in \meminitstateset.\:
  \neg \specfunc(\memstateargs{0}, \instr, \transrelmulti(\memstateargs{0}, \instr))$.
\end{definition}

Several approaches exist for single-instruction checking of a
processor, which is complementary to SQED (cf.~Section~\ref{sec:related:work}).


\section{Self-Consistency as QED-Consistency} \label{sec:qed:cons}

We now define QED-consistency
(cf.~Section~\ref{sec:sqed:overview}) as a property of states of a
processor $\transsys$ based on function $\funcmemlocscorr$. Then we formally
define the notion of QED test and show that for correct processors, QED tests
preserve QED-consistency.  This result is key to the proof of
the soundness in Section~\ref{sec:main:proofs} below.

\begin{definition}[QED-Consistency] \label{def:qed:consistency}
A state $\memstate$ is
\emph{QED-consistent}, written $\qedcons(\memstate)$, iff 
$\forall \memlocorig \in \funcmemlocsorig.\ \memstate(\memlocorig) =
\memstate(\funcmemlocscorr(\memlocorig))$.
\end{definition}

\noindent
QED-consistency is based on checking the architectural part of a
state. An equivalent condition can be formulated in
terms of duplicate locations:
$\forall \memlocdup \in \funcmemlocsdup.\ \memstate(\memlocdup) =
\memstate(\funcmemlocscorrinv(\memlocdup))$.

\begin{definition}[QED test] \label{def:canonical:seq}
An instruction sequence $\instrseq$
is a \emph{QED test} if $\instrseq = \instrorigseq ::
\funcdupseq(\instrorigseq)$ for some sequence $\instrorigseq$ of original
instructions.
\end{definition}
  
We link the abstract specification $\specfunc$ to the semantics of
original and duplicate instructions. This way, we obtain a notion of
functional congruence that readily follows as a special case from~\eqref{def:func:congr:instr}.

\begin{corollary}[Functional Congruence: Duplicate Instructions] \label{def:dup:specvalue}
Given $\instrorig \in \instrsetorig$ and $\instrdup \in \instrsetdup$
with $\instrdup = \funcdupsingle(\instrorig)$, the following holds
for all states $\memstate_0$, $\memstate_1$, $\memstate'$, and $\memstate''$:
\begin{align*}
 \big[  \specfunc( \memstate_0,  \instrorig, 
    \memstate') \wedge{} &
   \specfunc(\memstate_1, \instrdup, \memstate'') \wedge{}\\
     \memstate_0(\funcmemlocsin(\instrorig))
    & = \memstate_1(\funcmemlocscorr(\funcmemlocsin(\instrorig)))
    \big] \rightarrow \\
    & \memstate'(\funcmemlocsout(\instrorig))
    = \memstate''(\funcmemlocscorr(\funcmemlocsout(\instrorig))) 
\end{align*}
\end{corollary}

\noindent
Corollary~\ref{def:dup:specvalue} states
that an original instruction $\instrorig$ produces the same value at its output
location as its duplicate instruction $\instrdup =  \funcdupsingle(\instrorig)$,
provided that these instructions execute in states
where the values at the respective input locations match.

We generalize Corollary~\ref{def:dup:specvalue} to show that after executing a pair of original and
duplicate instructions, the values at \emph{all} original locations
match the values at the corresponding duplicate locations, assuming those
values also matched before executing the instructions.

\begin{lemma}[cf.~Corollary~\ref{def:dup:specvalue}] \label{lem:func:congr}
Given $\instrorig \in \instrsetorig$ and $\instrdup \in \instrsetdup$
with $\instrdup = \funcdupsingle(\instrorig)$, the following holds
for all states $\memstate_0$, $\memstate_1$, $\memstate'$, and $\memstate''$:
\begin{align*}
   \big[  \specfunc( \memstate_0, \instrorig, 
     \memstate') \wedge{} 
     \specfunc(\memstate_1, & \instrdup, \memstate'') \wedge{} \\
     \forall \memlocorig \in
    \funcmemlocsorig.\ \memstate_0(\memlocorig)
    & = \memstate_1(\funcmemlocscorr(\memlocorig))
    \big] \rightarrow \\
  & \forall \memlocorig \in
    \funcmemlocsorig.\ \memstate'(\memlocorig)
    = \memstate''(\funcmemlocscorr(\memlocorig)) 
\end{align*}
\end{lemma}
\begin{proof}
See appendix.\qedhere
\end{proof}

Lemma~\ref{lem:func:congr} leads to an important result that we need
to prove soundness of SQED
(Lemma~\ref{lem:bug:fail:qedtest:canonical:new:onlyif} below): executing a QED test
$\instrseq$ starting in a QED-consistent state results in a
QED-consistent final state if all instructions in $\instrseq$ execute
according to the abstract specification $\specfunc$ (cf.~Fig.~\ref{fig:sqed:qedcons}).

\begin{lemma}[QED-Consistency and QED tests] \label{lem:preserving:sep:seq:qed:cons}
Let $\instrseq = \langle \instrargs{1}, \ldots, \instrargs{2n} \rangle$ be a QED test,
let $\langle \memstateargs{0}, \ldots, \memstateargs{2n} \rangle$ be a
sequence of $2n+1$ states, and let $\specfunc$ be some abstract specification relation.  Then,
 \begin{align*}
    \qedcons(\memstate_0) \wedge
   \big( \bigwedge_{j := 0}^{2n-1} \specfunc(\memstate_{j}, \instrargs{j+1}, \memstate_{j+1}) & \big)
   \rightarrow \\ & \qedcons(\memstate_{2n})
 \end{align*}
\end{lemma}
\begin{proof}
Assuming the antecedent,
let $\memlocorig \in \funcmemlocsorig$ be arbitrary but fixed 
with $\memlocdup = \funcmemlocscorr(\memlocorig)$.  By repeated
application of~\eqref{cor:trans:spec:orig:instr}, we derive
$\memstateargs{0}(\memlocdup) = \memstateargs{1}(\memlocdup) = \ldots =
\memstateargs{n}(\memlocdup)$, and hence:
\begin{equation}
  \memstateargs{0}(\memlocdup) = \memstateargs{n}(\memlocdup) \label{eq:dup0-eq-dupn}
\end{equation}
by transitivity. By repeated
application of~\eqref{cor:trans:spec:dup:instr}, we derive:
\begin{equation}
  \memstateargs{n}(\memlocorig) = \memstateargs{2n}(\memlocorig) \label{eq:orgn-eq-org2n}
\end{equation}
Now, $\qedcons(\memstate_0)$ implies
$\memstateargs{0}(\memlocorig) = \memstateargs{0}(\funcmemlocscorr(\memlocorig))$,
from which it follows by \eqref{eq:dup0-eq-dupn} that
$\memstate_0(\memlocorig) = \memstateargs{n}(\funcmemlocscorr(\memlocorig))$.
By repeated application of Lemma~\ref{lem:func:congr}, we can next derive
$\memstate_j(\memlocorig) = \memstateargs{n+j}(\funcmemlocscorr(\memlocorig))$
for $1 \leq j \leq n$, and in particular, 
$\memstateargs{n}(\memlocorig) = \memstateargs{2n}(\funcmemlocscorr(\memlocorig))$.
Finally, by applying \eqref{eq:orgn-eq-org2n}, we get
$\memstateargs{2n}(\memlocorig) = \memstateargs{2n}(\funcmemlocscorr(\memlocorig))$.
Since $\memlocorig$ was chosen arbitrarily, $\qedcons(\memstateargs{2n})$ holds. \qedhere
\end{proof}


\section{Soundness and Conditional Completeness} \label{sec:main:proofs}

SQED checks a processor $\transsys$ for self-consistency by executing
QED tests and checking QED-consistency
(cf.~Fig~\ref{fig:sqed:overview}). We now define the correctness of
$\transsys$ in terms of QED tests that, when executed, always result
in QED-consistent states. This way, we establish a correspondence
between counterexamples to QED-consistency and bugs in $\transsys$. We
then prove our main results (Theorem~\ref{thm:main}) related to the bug-finding capabilities of
SQED, i.e., soundness and conditional completeness.

\begin{definition}[Failing and Succeeding QED Tests] \label{def:qed:test:failing}
Let $\instrseq$ be a QED test,
$\memstateargs{0}\in\meminitstateset$ an initial state such that
$\qedcons(\memstateargs{0})$ holds, and let
$\memstate=\transrelmulti(\memstateargs{0}, \instrseq)$.  We say that:
\begin{itemize}
\item QED test $\instrseq$ \emph{fails} if $\neg \qedcons(\memstate)$.
\item QED test $\instrseq$ \emph{succeeds} if $\qedcons(\memstate)$.
\end{itemize}
\end{definition}

\begin{definition}[Processor QED-Consistency] \label{def:proc:qed:cons}
  A processor $\transsys$ is \emph{QED-consistent} if all possible
  QED tests succeed.
\end{definition}

\begin{definition}[Processor QED-Inconsistency] \label{def:proc:qed:incons}
A processor $\transsys$ is \emph{QED-inconsistent} if some QED test fails.
\end{definition}

\begin{lemma} \label{lem:bug:fail:qedtest:canonical:new:onlyif}
Let $\transsys$ be a processor.  If
$\transsys$ is QED-inconsistent, then $\transsys$ is not correct with respect
to \emph{any} abstract specification relation.
\end{lemma}
\begin{proof}
Let $\instrseq$ be a failing QED test for $\transsys$ and 
assume
that processor $\transsys$ is correct with respect to some abstract
specification relation $\specfunc$.
By Lemma~\ref{lem:preserving:sep:seq:qed:cons}, we 
conclude $\qedcons(\memstateargs{2n})$, which contradicts the
assumption that $\instrseq$ is a failing QED \nolinebreak test.
\end{proof}

Importantly, Lemma~\ref{lem:bug:fail:qedtest:canonical:new:onlyif}
holds regardless of what the actual specification relation $\specfunc$
is, i.e., it is independent of $\specfunc$ and the opcode specification
function $\specfuncinstrname{\opcode}$ (Definition~\ref{def:spec}).

Lemma~\ref{lem:bug:fail:qedtest:canonical:new:onlyif} shows that SQED
is a \emph{sound} technique: any error reported by a failing QED test is in
fact a real bug in the system.  It is more challenging to determine the
degree to which SQED is \emph{complete}, that is, for which bugs do there exist
failing QED tests?  We address this question next.

Suppose that $\bug=\langle \instr_b, \triggerstatesnoarg \rangle$ is a bug
with respect to a specification $\specfunc$ in a
processor $\transsys$, where $\instr_b = (\opcode_b, \memloc^b_{out},
(\memloc^b_{in1},\memloc^b_{in2}))$.  A \emph{bug-specific QED test} for $\bug$ is a QED test that
sets up the conditions for and includes the activation of the bug.  
By 
Definition~\ref{def:bug:pair}, if $\instr_b$ is executed in
$\transsys$ starting from any state in $\triggerstatesnoarg$, the
specification is violated.  That is, for each $\memstate_b \in
\triggerstatesnoarg$, 
$\neg \specfunc(\memstate_b, \instr_b, \transrel(\memstate_b, \instr_b))$.
Let $\memstate=\transrel(\memstate_b, \instr_b)$. 
According to~\eqref{assume:trans:spec}, there are two ways the specification can be violated.  Either:
(A) the value in the output location of $\instr_b$ is different from that
required by $\specfunc$, i.e.: $\memstate(\memloc^b_{out}) \not =
\specfuncinstr{\opcode_b}{\memstate_b(\memloc^b_{in1}),\memstate_b(\memloc^b_{in2})}$, 
which we call a \emph{type-A bug};
or (B) the value in some other, non-output location $\memloc_{\mathit{bad}}$ is not preserved, i.e.:
$\memstate(\memloc_{\mathit{bad}}) \not= \memstate_b(\memloc_{\mathit{bad}})$ for some $\memloc_{\mathit{bad}} \not=
\memloc^b_{out}$, which we call a \emph{type-B bug}.
We now define a bug-specific QED test formally. 

\begin{definition}[Bug-Specific QED Test] \label{def:bug:specific:qed:test}
Let $\bug = \langle \instr_b, \triggerstatesnoarg \rangle$ be a bug in
$\transsys$ with respect to $\specfunc$, where $\instr_b =
(\opcode_b, \memloc^b_{out}, (\memloc^b_{in1},\memloc^b_{in2}))$.
The instruction sequence $\instrseq = \langle
\instrargs{1}, \ldots, \instrargs{n}, \instrargs{n+1}, \ldots, \instrargs{2n}\rangle$
is a \emph{bug-specific QED test} for $\bug$ if the following conditions
hold:
\begin{enumerate}

\item \label{define:qed:second:bug} $\instrargs{n+1} = \instr_b$.

\item \label{define:qed:fourth:dup} $\instrseq$ is a QED test for some
  $\funcmemlocscorr$, i.e. for $1 \le k \le n$, $\instrargs{n+k} =
  \funcdupsingle(\instrargs{k})$.  In particular, $\instrargs{1} = (\opcode_b, \memloc_{out},
  (\memloc_{in1},\memloc_{in2}))$, with $(\memloc_{in1},\memloc_{in2},\memloc_{out}) = \funcmemlocscorr^{-1}(
  (\memloc^b_{in1},\memloc^b_{in2},\memloc^b_{out}))$.

\item \label{define:qed:third:path} There exists a path $\statepath \in \memstateset^{2n}$ from
  $\memstateargs{0} \in \meminitstateset$ with $\qedcons(\memstateargs{0})$, such that
  $\statepath =
  \transrelmulti(\memstateargs{0}, \instrseq) = \langle
  \memstateargs{1}, \ldots, \memstateargs{n}, \memstateargs{n+1}, \ldots,
  \memstateargs{2n}\rangle$, where $\memstateargs{n} \in \triggerstatesnoarg$.

\item $\specfunc(s_0,\instr_1,s_1)$.

\item \label{define:qed:fourth:cases} Additionally, we need three more conditions that depend on the bug types: 
  \begin{itemize}
    \item[Case A:] If $\instr_b$ is a type-A bug with respect to $\memstateargs{n}$, i.e.
$\memstateargs{n+1}(\memloc^b_{out}) \not =
  \specfuncinstr{\opcode_b}{\memstateargs{n}(\memloc^b_{in1}),\memstateargs{n}(\memloc^b_{in2})}$,
  then let $\memlocinv = \memloc_{out}$ and $\memlocinvy = \memloc^b_{out}$.
\item We then require:
  \begin{itemize}
  \item $\memstateargs{n+1}(\memlocinvy) = \memstateargs{2n}(\memlocinvy)$,
  \item $\memstateargs{1}(\memlocinv) = \memstateargs{2n}(\memlocinv)$,
  \item $\memstateargs{0}(\funcmemlocsin(\instr_b)) =
    \memstateargs{n}(\funcmemlocsin(\instr_b))$.
  \end{itemize}

\item[Case B:] If $\instr_b$ is a type-B bug with respect to $\memstateargs{n}$, i.e.
$\memstateargs{n}(\memloc_{\mathit{bad}}) \not = \memstateargs{n+1}(\memloc_{\mathit{bad}})$ for
  some $\memloc_{\mathit{bad}}\not=\memloc^b_{out}$, then let $\memlocinv =
  \funcmemlocscorr^{-1}(\memloc_{\mathit{bad}})$ with $\memlocinv \not = \memloc_{out}$ and $\memlocinvy = \memloc_{\mathit{bad}}$.
\item We then require:
  \begin{itemize}

  \item $\memstateargs{n+1}(\memlocinvy) = \memstateargs{2n}(\memlocinvy)$,
    
  \item $\memstateargs{1}(\memlocinv) = \memstateargs{2n}(\memlocinv)$.

  \item $\memstateargs{1}(\memlocinvy) = \memstateargs{n}(\memlocinvy)$,

  \end{itemize}
\end{itemize}

\end{enumerate}
\end{definition}

Clearly, it is always possible to satisfy the first two conditions by
declaring the buggy instruction $\instr_b$ to be the duplicate of
$\instrargs{1}$ with respect to some function
$\funcmemlocscorr$.  Moreover,
if we restrict our attention to single-instruction correct processors, then the
fourth condition always holds as well.  This fits in well with the stated intended
role of SQED which is to find sequence-dependent bugs, rather than
single-instruction bugs.

Understanding when the remaining conditions~\ref{define:qed:third:path} and~\ref{define:qed:fourth:cases} hold is more complicated.
We must find some instruction sequence $\instrseq^*=\langle
\instrargs{2}\ldots\instrargs{n}\rangle$ that can transition $\transsys$ 
from the state $\memstateargs{1}$ following the execution of $\instr_1$ to one of the
bug-triggering states in $\triggerstatesnoarg$, i.e., $\memstateargs{n}$.  Often it is
reasonable to assume that $\transsys$ is \emph{strongly connected}, i.e.,
that there always exists an instruction sequence that can transition from
one reachable state to another.
This is almost enough to ensure the existence
of $\instrseq^*$.  However, there are a few other restrictions on $\instrseq^*$ 
to satisfy Definition~\ref{def:bug:specific:qed:test}.

First, $\instrseq^*$ must consist of only original instructions to satisfy the definition of a QED test.  We are free to choose
$\funcmemlocscorr$ to be anything that works, so the main restriction is that
$\instrseq^*$ cannot use any instructions referencing locations that are used by $\instr_b$, i.e., 
$\memloc^b_{in1}$, $\memloc^b_{in2}$, or $\memloc^b_{out}$. Note that
we defined $\instrargs{n+1} = \instr_b$ to be the first duplicate instruction. 
This ends up being the most severe restriction on $\instrseq^*$ because it
means that instructions in $\instrseq^*$ cannot write to the locations used as
inputs by $\instr_b$.  We discuss some mitigations to this restriction in Section~\ref{sec:extensions}.

Somewhat surprisingly, the three requirements in condition~\ref{define:qed:fourth:cases}
are not very severe, as we now explain.
For both type-A and type-B bugs, locations $\memlocinv$ and $\memlocinvy$ are an original location and
its duplicate, respectively, that will hold inconsistent values when the QED test $\instrseq$ fails.
For type-A bugs, $\memlocinv$ holds the
correct output value of $\instrargs{1}$ and $\memlocinvy$ holds the
incorrect output value of $\instr_b$.  For type-B bugs, $\memlocinvy$
holds the value of location $\memloc_{\mathit{bad}}$ that is
incorrectly modified when $\instr_b$ is executed in state
$\memstateargs{n}$, and $\memlocinv$ is the original
location that corresponds to $\memlocinvy = \memloc_{\mathit{bad}}$. 

The first requirement $\memstateargs{n+1}(\memlocinvy) =
\memstateargs{2n}(\memlocinvy)$ means that the duplicate
sequence $\funcdupseq(\instrseq^*)$ of $\instrseq^*$ in the QED test
has to preserve the value of $\memlocinvy$ in $\memstateargs{n+1}$
also in the final state $\memstateargs{2n}$. Further, since
$\memlocinv = \funcmemlocscorr^{-1}(\memlocinvy)$, this also imposes
restrictions on the modifications that $\instrseq^*$ can make to
$\memlocinv$. However, as this is just one original location, it is 
unlikely that every possible $\instrseq^*$ 
would need to modify it to get to some bug-triggering state \nolinebreak $\memstateargs{n}$.

The second requirement is $\memstateargs{1}(\memlocinv) =
\memstateargs{2n}(\memlocinv)$.  For similar reasons, it is unlikely that $\instrseq^*$ would need to modify
$\memlocinv$, and the duplicate sequence $\funcdupseq(\instrseq^*)$ of
$\instrseq^*$ should not modify it either, since it is an original
location and original locations should be left alone by duplicate
instructions. Although the buggy instruction $\instr_b$ might modify
$\memlocinv$ if it has more than one bug effect, we may be able to
choose the locations of $\instr_1$ and $\funcmemlocscorr$ differently
to avoid this.

Finally, the last requirement of
condition~\ref{define:qed:fourth:cases} depends on the
two cases A and B. In both cases, we require that $\instrseq^*$ does
not modify certain duplicate locations: the input locations
$\funcmemlocsin(\instr_b)$ of $\instr_b$ (A) and 
location $\memlocinvy$ that is incorrectly modified by 
$\instr_b$ (B).  Sequence $\instrseq^*$ should not modify any duplicate
locations as it is composed of original instructions.
Note that we do not have to make the
strong assumption that $\instrseq^*$ 
executes according to its specification, only that it avoids corrupting
a few key locations. Given that we have a lot of freedom in choosing
$\funcmemlocscorr$ and hence the locations of $\instr_1$, these requirements are likely to be
satisfiable if there are some degrees of freedom in choosing a path to one
of the bug-triggering states.

We now prove our conditional completeness property, namely that if a
bug-specific QED test $\instrseq$ exists, then $\instrseq$ fails.

\begin{lemma} \label{lem:bug:fail:qedtest:canonical:new:if}
Let $\transsys$ be a processor with a bug $\bug = \langle \instr_b, \triggerstatesnoarg
\rangle$ with respect to specification $\specfunc$, for which there
exists a bug-specific QED test $\instrseq$.  Then $\instrseq$ fails.
\end{lemma}
\begin{proof}
Let $\bug = \langle \instr_b, \triggerstatesnoarg \rangle$ be a bug and
$\instrseq$ be a bug-specific QED test for $\bug$. 
By Definition~\ref{def:bug:specific:qed:test} we have $\instrseq = \langle
\instrargs{1}, \ldots, \instrargs{n}, \instrargs{n+1}, \ldots, \instrargs{2n}\rangle$ and $\statepath =
  \transrelmulti(\memstateargs{0}, \instrseq) = \langle
  \memstateargs{0}, \memstateargs{1}, \ldots, \memstateargs{n}, \memstateargs{n+1}, \ldots,
  \memstateargs{2n}\rangle$, where $\memstateargs{n} \in \triggerstatesnoarg$
  and $\instr_b = \instrargs{n+1}$,
  and $\qedcons(\memstateargs{0})$ holds. We show that $\neg
  \qedcons(\memstateargs{2n})$ holds by showing that $\memstateargs{2n}(\memlocinv)
    \not = \memstateargs{2n}(\memlocinvy)$. We distinguish the two cases A and B in
  Definition~\ref{def:bug:specific:qed:test}.

  \noindent \textbf{Case A.}     
    Since $\qedcons(\memstateargs{0})$ and
    $\funcdupsingle(\instrargs{1}) = \instr_b$, we have
    \begin{equation}
      \memstateargs{0}(\funcmemlocsin(\instrargs{1})) =
      \memstateargs{0}(\funcmemlocsin(\instr_b)) \label{case:A:line0}
    \end{equation} 
    From the third requirement of Case A in Definition~\ref{def:bug:specific:qed:test}, we have
 $\memstateargs{0}(\funcmemlocsin(\instr_b)) =
    \memstateargs{n}(\funcmemlocsin(\instr_b))$, so it follows that,
    \begin{equation}
    \memstateargs{0}(\funcmemlocsin(\instrargs{1})) =
    \memstateargs{n}(\funcmemlocsin(\instr_b)) \label{case:A:line1}
    \end{equation}
    By \eqref{case:A:line1} and since $\funcopcode(\instrargs{1}) = \funcopcode(\instr_b)$, also
    \begin{equation}
    \specfuncinstr{\funcopcode(\instrargs{1})}{\memstateargs{0}(\funcmemlocsin(\instrargs{1}))}
    =
    \specfuncinstr{\funcopcode(\instr_b)}{\memstateargs{n}(\funcmemlocsin(\instr_b))} \label{case:A:line2}
    \end{equation} 
    Since $\specfunc(s_0,\instr_1,s_1)$ by Definition~\ref{def:bug:specific:qed:test}, we have
    \begin{equation}
    \memstateargs{1}(\funcmemlocsout(\instrargs{1})) =
    \specfuncinstr{\funcopcode(\instrargs{1})}{\memstateargs{0}(\funcmemlocsin(\instrargs{1}))} \label{case:A:line3}
    \end{equation} 
    Since we are in Case A, we have from
    Definition~\ref{def:bug:specific:qed:test} that
    $\memlocinv = \funcmemlocsout(\instrargs{1})$, and from the second
    requirement of Case A, we have $\memstateargs{1}(\memlocinv) =
      \memstateargs{2n}(\memlocinv)$, so it follows that,
    \begin{equation}
    \memstateargs{2n}(\memlocinv) =
    \specfuncinstr{\funcopcode(\instrargs{1})}{\memstateargs{0}(\funcmemlocsin(\instrargs{1}))} \label{case:A:line5}
    \end{equation} 
    Since $\instr_b$ fails in state
    $\memstateargs{n}$, we have that,
    \begin{equation}
    \memstateargs{n+1}(\funcmemlocsout(\instr_b)) \not =
    \specfuncinstr{\funcopcode(\instr_b)}{\memstateargs{n}(\funcmemlocsin(\instr_b))} \label{case:A:line6}
    \end{equation} 
    Again, from Case A in
    Definition~\ref{def:bug:specific:qed:test}, we have
    $\memlocinvy = \funcmemlocsout(\instr_b)$,
    and from the first requirement of Case A, we
    have $\memstateargs{n+1}(\memlocinvy) =
    \memstateargs{2n}(\memlocinvy)$, so it follows that,
    \begin{equation}
    \memstateargs{2n}(\memlocinvy) \not =
    \specfuncinstr{\funcopcode(\instr_b)}{\memstateargs{n}(\funcmemlocsin(\instr_b))} \label{case:A:line7}
    \end{equation} 
    Finally, \eqref{case:A:line2} and \eqref{case:A:line5} give us,
    \begin{equation}
      \memstateargs{2n}(\memlocinv) = \specfuncinstr{\funcopcode(\instr_b)}{\memstateargs{n}(\funcmemlocsin(\instr_b))} \label{case:A:line8}
    \end{equation} 
     But then \eqref{case:A:line7} and \eqref{case:A:line8} imply  $\memstateargs{2n}(\memlocinv)
    \not = \memstateargs{2n}(\memlocinvy)$, and hence $\neg \qedcons(\memstateargs{2n})$.  

  \noindent \textbf{Case B.} See appendix.     
    \qedhere
\end{proof}

\begin{theorem} \label{thm:main}
  \  

  \begin{itemize}
\item SQED is sound (Lemma~\ref{lem:bug:fail:qedtest:canonical:new:onlyif}).
\item SQED is complete for bugs for which a bug-specific QED test exists
(Lemma~\ref{lem:bug:fail:qedtest:canonical:new:if}).
\end{itemize}
\end{theorem}

Theorem~\ref{thm:main} is relevant for practical applications of SQED.
Referring to the high-level workflow shown in
Fig.~\ref{fig:sqed:overview},
BMC symbolically explores all possible QED tests up to bound $n$ for a
particular fixed mapping $\funcmemlocscorr$. If a failing QED test
$\instrseq$ is found, then by the soundness of SQED, $\instrseq$
corresponds to a bug in the processor. By completeness, if there
exists a bug for which a bug-specific QED test $\instrseq$ exists,
then with a sufficiently large bound $n$, BMC will find a sequence $\instrseq$
that will fail.

\subsection{Extensions}
\label{sec:extensions}

We now consider variants of QED tests that cover a larger class of
bugs (i.e. bugs that cannot be detected by a bug-specific QED
test). Ultimately, with hardware support we obtain a family
of QED tests which, together with single-instruction correctness,
results in a complete variant of SQED (Theorem~\ref{thm:hard:reset}).

The main limitation of bug-specific QED
tests arises from
the fact that QED tests consist of a sequence of original instructions
followed by duplicate ones. This makes it impossible to set up a
bug-specific QED test for an important class of forwarding-logic bugs
(a simple refinement of our model can be used for the important case of
pipelined systems).  To see why, consider that a
bug-triggering state $\memstateargs{n} \in \triggerstatesnoarg$ must
be reached by executing a sequence of original instructions. The
buggy instruction, which is a \emph{duplicate}, is executed in state
$\memstateargs{n}$ and would have to read a value from some
\emph{original} location written previously.

To resolve this limitation, first note that there is another way that SQED can find bugs, namely by finding
QED tests for which the bug occurs during the original sequence, but not during the duplicate one.
This kind of QED test is much more effective with a simple extension
to allow no-operation instructions (a trick also employed
in~\cite{DBLP:conf/fmcad/JonesSD96}).
To formalize this, we first define a set
$\mathcal{N}$ of no-operation instructions (NOPs).

\begin{definition}
Let $\mathcal{N}$ be the set of instructions such that, for every state
$(\memstatearch,\memstatenarch)$, if $\instrnop \in \mathcal{N}$, then
$\transrel((\memstatearch,\memstatenarch), \instrnop) = (\memstatearch,\memstatenarch')$ for
some $\memstatenarch' \in \memstatesetnarch$.
\end{definition}

\noindent
An instruction in $\mathcal{N}$ may change the non-architectural part of a state, but
not the architectural part.

\begin{definition}
An \emph{extended QED test} is any sequence of instructions obtained from a standard
QED test by inserting zero or more instructions from $\mathcal{N}$ anywhere in
the sequence.
\end{definition}

\noindent
Extended QED tests enjoy the same properties as standard
QED tests.  In particular, an appropriately lifted version of Lemma~\ref{lem:preserving:sep:seq:qed:cons} holds and
the notions of failing and succeeding QED tests can be lifted to extended QED
tests in the obvious way.

\begin{definition}[Bug-Hunting Extended QED Test] \label{def:bug:specific:extended:qed:test}
Let $\transsys$ be a single-instruction correct processor with at least one
bug.  The instruction sequence $\instrseq$ is a \emph{bug-hunting extended QED test}
with a \emph{bug-prefix of size $k$} and initial state $\memstateargs{0}$ for $\transsys$ if the following conditions hold:
\begin{enumerate}
  \item There is some bug $\bug = \langle \instr_b, \triggerstatesnoarg \rangle$ in
    $\transsys$ such that $\transrelmulti(\memstateargs{0}, \langle \instrargs{1},\ldots,\instrargs{k-1}\rangle)
    \in \triggerstatesnoarg$ and $\instrargs{k} = \instr_b$
  \item $\instrseq$ is an extended QED test
  \item $\instrargs{k}$ is an original instruction, and $\instrargs{k+1} = \funcdupsingle(\instrargs{1})$
\end{enumerate}
\end{definition}

\noindent
Unlike a bug-specific QED test, a bug-hunting extended QED test
is not guaranteed to fail.  It starts with a bug-triggering
sequence of length $k$, and then finishes with a modified duplicate sequence which
may add (or subtract) NOPs from $\mathcal{N}$.  The
NOPs can be used to change the timing between any interdependent instructions,
making it more likely that the duplicate sequence will produce a correct
result, especially if the bug depends on forwarding-logic. 
One can show (omitted for lack of space) that for a general class of forwarding-logic
bugs, there does always exist an extended QED test that fails.

Another QED test extension is to allow original and duplicate instructions to be
\emph{interleaved}~\cite{DBLP:conf/iccad/LonsingGMNSSYMB19}, rather than requiring that all original instructions
precede all duplicate instructions~\cite{8355908}.\footnote{
The bug in Example~\ref{ex:bug:detection} can be detected by executing
the QED test $\instrseq = \instrorigargs{1},\instrdupargs{1} ::
\instrorigargs{2},\instrdupargs{2}$, which interleaves original and
duplicate instructions. The subsequence
$\instrorigargs{2},\instrdupargs{2}$ of two back-to-back MULs causes
$\instrdupargs{2}$ to produce an incorrect result at its output
location $\memloc_{31}$.  The final state is QED-inconsistent since
the output location $\memloc_{15}$ of $\instrorigargs{2}$ holds the
correct value, while $\memloc_{31}$ holds an incorrect one.}
Again, it is straightforward to show that
this extension preserves Lemma~\ref{lem:preserving:sep:seq:qed:cons}.
Clearly, the set of bugs that can be found by adding interleaving are
a strict superset of those that can be found without.
In practice, implementations of SQED search for all possible extended QED
tests with interleaving.  Empirically, case studies have not turned up any
(non-single-instruction) bugs that cannot be found with this combination.
However, one can construct pathological systems with bugs that cannot be found
by such QED tests. We address these cases next.

\subsection{Hardware Extensions}
With hardware support, stronger guarantees can be achieved
that lead to our final completeness result
(Theorem~\ref{thm:hard:reset}).  We first introduce a \emph{soft-reset}
instruction, which transitions the non-architectural part of a state to the
initial non-architectural state $\meminitstatenarch$ without changing the architectural
part. Then we define a variant of bug-hunting
extended QED tests where we insert soft-reset instructions in the
sequence of duplicate instructions. This way, all duplicate
instructions execute in an initial state and hence execute according
to the specification for single-instruction correct processors. The
resulting QED test always fails, in contrast to a bug-hunting extended
QED \nolinebreak test.

\begin{definition}
$\instr_r$  is a \emph{soft-reset instruction} for $\transsys$ if for every state
$(\memstatearch,\memstatenarch)$, $\transrel((\memstatearch,\memstatenarch), i_r) =
  (\memstatearch,\meminitstatenarch)$.
\end{definition}

\noindent
It is easy to see that $\instr_r \in \mathcal{N}$.

\begin{definition}[Bug-Specific Soft-Reset QED Test] \label{def:bug:specific:soft-reset-qed:test}
Let $\transsys$ be single-instruction correct with at least one
bug $\bug = \langle \instr_b, \triggerstatesnoarg \rangle$.
The instruction sequence $\instrseq = \langle \instrargs{1}, \ldots \instrargs{n}\rangle$
is a \emph{bug-specific soft-reset QED test} for $\transsys$ if the following conditions
hold:
\begin{enumerate}
  \item $\instrseq$ is a bug-hunting extended QED test for $\transsys$ with a minimal
    bug-prefix of size $k\ge 2$ and initial state $\memstateargs{0}$
  \item Let $\statepath=\transrelmulti(\memstateargs{0},\instrseq)$. Then,
    $\forall\,\memloc \in \funcmemlocsdup.\: \memstateargs{k-1}(\memloc) =
    \memstateargs{k}(\memloc)$, i.e., $\instr_b = \instrargs{k}$ does not corrupt any
    duplicate location
  \item $n = 3k$
  \item For each $1\le j \le k$, $\instr_{k+2j-1} = \instr_r$
\end{enumerate}
\end{definition}

\begin{lemma} \label{lem:bug:specific:soft-reset-qed:test}
  If $\transsys$ is single-instruction correct and has a bug-specific soft-reset QED test $\instrseq$, then
  $\instrseq$ fails.
\end{lemma}
\begin{proof}
See appendix. \qedhere
\end{proof}

There are still a few (pathological) ways in which a bug may be missed by searching for
all possible soft-reset QED tests.  First, there may be no triggering sequence
starting from any QED-consistent state.  Second, it could be that the
triggering sequence for a bug requires using more than half of all the
locations, making it impossible to divide the locations among original and
duplicate instructions.  Finally, it could be that the bug always corrupts
duplicate locations for every possible candidate sequence.  These can all be
remedied by adding \emph{hard reset} instructions, which 
reset $\transsys$ to a specific initial state.

\begin{definition}
The set $\{\instr_{R,\meminitstate}|\meminitstate \in \meminitstateset\}$ is a
\emph{family of hard reset instructions} for $\transsys$ if for every state
$\memstate$, $\transrel(\memstate,\instr_{R,\meminitstate}) = \meminitstate$.
\end{definition}

\begin{definition}
Let $\transsys$ be a processor.
Then $\instr = \langle \instrargs{1}\ldots\instrargs{2k+2}\rangle$ is a
\emph{bug-specific hard-reset QED test} with bug-prefix size $k$ and initial state
$\meminitstate$ for $\transsys$ if the following conditions hold:
\begin{enumerate}
  \item $k \ge 2$
  \item $\langle \instr_1\ldots\instr_k \rangle$ reach and trigger a bug $\bug = \langle \instr_b, \triggerstatesnoarg \rangle$ in
    $\transsys$ starting from $\meminitstate$, where $\instr_k = \instr_b$ 
  \item $\instr_{k+1} = \instr_{R,\meminitstate}$
  \item $\langle \instrargs{k+2} \ldots \instrargs{2k} \rangle = \langle \instrargs{1} \ldots \instrargs{k-1}\rangle$
  \item $\instrargs{2k+1} = \instr_r$
  \item $\instrargs{2k+2} = \instrargs{k}$
\end{enumerate}
\end{definition}

\noindent
Notice that there is no notion of duplication for a hard-reset QED test.
Instead, the exact same sequence is executed twice except that there is a hard
reset in between and a soft reset right before the last instruction.
Hard-reset QED tests also use a slightly different notion of success and failure.

\begin{definition}
Let $\instrseq$ be a bug-specific hard-reset QED test with bug-prefix size $k$ and initial
state $\meminitstate$, and let $\statepath = \transrelmulti(\meminitstate, \instrseq)$.
\begin{itemize}
  \item $\instrseq$ succeeds if $\memstateargs{k}(l) = \memstateargs{2k+2}(l)$
    for every location $l \in \memlocs$.
  \item $\instrseq$ fails if $\memstateargs{k}(l) \not= \memstateargs{2k+2}(l)$
    for some location $l \in \memlocs$.
\end{itemize}
\end{definition}

\noindent
The combination of single-instruction correctness checking and exhaustive search
for hard-reset QED tests is complete.

\begin{theorem} \label{thm:hard:reset}
If $\transsys$ is single-instruction correct and has no failing
bug-specific hard-reset QED tests, then it is correct.
\end{theorem}
\begin{proof}
See appendix. \qedhere
\end{proof}


\section{Related Work} \label{sec:related:work}

Assertion-based formal verification techniques using theorem proving
or (bounded) model checking, e.g.,
\cite{DBLP:journals/jar/Hunt89,DBLP:conf/cav/BurchD94,DBLP:conf/tacas/BiereCCZ99,DBLP:conf/cav/BiereCRZ99},
require implementation-specific, manually-written properties. In contrast to
that, \emph{symbolic quick error detection
(SQED)}~\cite{DBLP:conf/itc/LinSBM15,8355908,DBLP:conf/date/SinghDSSGFSKBEM19,DBLP:conf/iccad/LonsingGMNSSYMB19}
is based on a universal self-consistency property.

In an early application of self-consistency checking for processor
verification without a specification~\cite{DBLP:conf/fmcad/JonesSD96},
given instruction sequences are transformed by, e.g., inserting 
NOPs. The original and the modified instruction sequence are expected
to produce the same result. As a formal foundation, this approach
relies on formulating and explicitly computing an equivalence relation over 
states, which is not needed with \nolinebreak SQED. 

SQED originates from \emph{quick error detection (QED)}, a
post-silicon validation
technique~\cite{DBLP:conf/itc/HongLPMLKHNGM10,DBLP:conf/date/LinHLFGHM13,DBLP:journals/tcad/LinHLSKFHGM14}.
QED is highly effective in reducing the length of existing bug traces
(i.e., instruction sequences) in post-silicon debugging of processor
cores. To this end, existing bug traces are
systematically transformed into \emph{QED tests} by techniques that
(among others) include instruction
duplication~\cite{DBLP:journals/tr/OhSM02}. SQED exhaustively searches
for minimal-length QED tests using BMC for pre-silicon
verification. It is also applicable to post-silicon validation. SQED
was extended to operate with symbolic initial
states~\cite{DBLP:conf/date/FadihehUNMBSK18,CS2QEDDATE2020} to
overcome the potential limitations of BMC when unrolling the
transition relation of a design starting in a concrete initial state.

SQED employs the principle of self-consistency based on a mathematical
interpretation of instructions as functions. That principle is also
applied by \emph{accelerator quick error detection
  (A-QED)}~\cite{AQEDDAC2020accepted}, a formal pre-silicon
verification technique for HW accelerator designs.  A-QED checks the
functions implemented by an accelerator for functional consistency
and, like SQED, does not require a formal specification.

\emph{Unique program execution
  checking}~\cite{DBLP:conf/date/FadihehSBMK19} relies on a particular
variant of self-consistency to check security vulnerabilities of
processor designs for covert channel attacks.  In the context
of security, self-consistency is also applied to verify
secure information flow by self-composition of
programs~\cite{DBLP:conf/csfw/BartheDR04,DBLP:conf/fm/BartheCK11,DBLP:conf/uss/AlmeidaBBDE16,DBLP:conf/cav/YangVSGM18}.

Several approaches, including both formal and simulation-based
approaches, exist for checking \emph{single-instruction (SI)
  correctness}
cf.~\cite{DBLP:conf/cav/ReidCDGHKPSVZ16,DBLP:conf/date/SinghDSSGFSKBEM19,CS2QEDDATE2020}. Checking
SI correctness is complementary to checking self-consistency using
SQED and is also much more tractable.  In a formal approach, a
property corresponding to $\specfuncinstrname{\opcode}$ (based on the
ISA) is written for each opcode $\opcode \in \opcodeset$, and the
model checker is used to ensure that the property holds when starting
from any initial state.  Because the approach is restricted to initial
states and only a single instruction execution, it is much simpler to
specify and check than would be a property specifying the full correctness
of $\transsys$.  Efficient specialized approaches exist for checking
multiplier
units~\cite{DBLP:conf/aspdac/KrautzWKWJP08,DBLP:conf/date/Sayed-AhmedGKSD16,DBLP:conf/fmcad/RitircBK17,DBLP:conf/fmcad/KaufmannBK19},
which is computationally hard.


\section{Conclusion and Future Work} \label{sec:conclusion}

We laid a formal foundation for symbolic quick error detection (SQED) and presented a theoretical
framework to reason about its bug-finding capabilities. In our
framework, we proved soundness as well as (conditional) completeness, thereby
closing a gap in the theoretical understanding of SQED. Soundness
implies that SQED does not produce spurious counterexamples,
i.e., any counterexample to QED-consistency reported by SQED
corresponds to an actual bug in the design. For completeness, we
characterized a large class of bugs that can be detected by failing
QED tests under modest assumptions about these bugs.  We also identified
several QED test extensions based on executing no-operation
and reset instructions. For these extensions, we proved
even stronger completeness guarantees, ultimately leading to a variant
of SQED that, together with single-instruction correctness, is \nolinebreak complete.

As future work, it would be valuable to extend our framework to consider
variants of SQED that operate with more fully symbolic initial
states~\cite{DBLP:conf/date/FadihehUNMBSK18,CS2QEDDATE2020}. The challenge will
be to identify how this can be done while guaranteeing no spurious
counterexamples.  For practical applications,
our theoretical results provide valuable insights. For example, in
present implementations of
SQED~\cite{DBLP:conf/date/SinghDSSGFSKBEM19,DBLP:conf/iccad/LonsingGMNSSYMB19},
the flexibility to partition register/memory locations into sets of
original and duplicate locations and to select the bijective mapping between
these two sets has not yet been explored. Similarly, it is promising to
combine standard QED tests and the specialized extensions we presented
in a uniform practical tool framework. Features like
soft/hard reset instructions could either be implemented in HW in a
design-for-verification approach or in software inside a model
checker. In another research direction, we plan to extend our
framework to model the detection of deadlocks using SQED,
cf.~\cite{DBLP:conf/itc/LinSBM15}, and prove related theoretical
guarantees.

\textbf{Acknowledgments.} We thank Karthik Ganesan and John Tigar
Humphries for helpful initial discussions and the anonymous reviewers for their feedback.


\IEEEtriggeratref{19}



\newpage

\begin{appendices}

\section{Proofs}

\begin{proof}[Proof of Lemma~\ref{lem:func:congr}]
Assume that the antecedent of the implication holds, and
let $\memlocorig \in \funcmemlocsorig$ be an arbitrary
original memory location.  If $\memlocorig = \funcmemlocsout(\instrorig)$,
then $\memstate'(\funcmemlocsout(\instrorig)) =
\memstate''(\funcmemlocscorr(\funcmemlocsout(\instrorig)))$ by
Corollary~\ref{def:dup:specvalue}.

Suppose, on the other hand, that $\memlocorig \not=
\funcmemlocsout(\instrorig)$.  Let $\memlocdup = \funcmemlocscorr(\memlocorig)$ be the corresponding duplicate
location. By the injectivity of
$\funcmemlocscorr$, we have $\memlocdup \not =
\funcmemlocscorr(\funcmemlocsout(\instrorig))$, and thus
$\memlocdup \not= \funcmemlocsout(\instrdup)$.
We can thus conclude from~\eqref{assume:trans:spec} that $\memstate_0(\memlocorig)
= \memstate'(\memlocorig)$ and $\memstate_1(\memlocdup) =
\memstate''(\memlocdup)$.

Finally, since $\memstate_0(\memlocorig) =
\memstate_1(\memlocdup)$ by assumption, we derive $\memstate'(\memlocorig) =
\memstate''(\memlocdup)$, that is, $\memstate'(\memlocorig) =
\memstate''(\funcmemlocscorr(\memlocorig))$.
\qedhere
\end{proof}

\begin{proof}[Proof of Case B of Lemma~\ref{lem:bug:fail:qedtest:canonical:new:if}]
  Since $\qedcons(\memstateargs{0})$, we have
    \begin{equation}
    \memstateargs{0}(\memlocinv) = \memstateargs{0}(\memlocinvy) \label{case:B:trans:init}
      \end{equation}         
  Since $\specfunc(s_0,\instr_1,s_1)$ by Definition~\ref{def:bug:specific:qed:test}, we have
    $\memstateargs{0}(\memlocinv) = \memstateargs{1}(\memlocinv)$ and
    $\memstateargs{0}(\memlocinvy) = \memstateargs{1}(\memlocinvy)$, and so it
  follows that,
    \begin{equation}
    \memstateargs{1}(\memlocinv) =
    \memstateargs{1}(\memlocinvy) \label{case:B:trans0}
    \end{equation} 
  Due to the requirements in Case B of Definition~\ref{def:bug:specific:qed:test}, we have
    \begin{equation}
      \memstateargs{n+1}(\memlocinvy) = \memstateargs{2n}(\memlocinvy) \label{case:B:cond:0}
    \end{equation} 
    \begin{equation}
    \memstateargs{1}(\memlocinv) = \memstateargs{2n}(\memlocinv) \label{case:B:cond:1}
    \end{equation}
    \begin{equation}
     \memstateargs{1}(\memlocinvy)
     = \memstateargs{n}(\memlocinvy) \label{case:B:cond:2}
    \end{equation}
  Now, because we are in Case B, we know that $\memstateargs{n}(\memlocinvy) \not =
    \memstateargs{n+1}(\memlocinvy)$, so by \eqref{case:B:cond:0},
    \begin{equation}
    \memstateargs{2n}(\memlocinvy) \not = \memstateargs{n}(\memlocinvy) \label{case:B:trans1}
    \end{equation} 
  But \eqref{case:B:trans0}, \eqref{case:B:cond:1} and \eqref{case:B:cond:2} give us:
    \begin{equation}
    \memstateargs{n}(\memlocinvy) = \memstateargs{2n}(\memlocinv) \label{case:B:trans2}
    \end{equation} 
  Thus, by \eqref{case:B:trans1} and \eqref{case:B:trans2},
    \begin{equation}
    \memstateargs{2n}(\memlocinvy) \not = \memstateargs{2n}(\memlocinv)
    \end{equation} 
  and hence $\neg \qedcons(\memstateargs{2n})$. \qedhere
\end{proof}

\begin{proof}[Proof of Lemma~\ref{lem:bug:specific:soft-reset-qed:test}]
  Because $k$ is minimal, we know that $\instrargs{1}\ldots\instrargs{k-1}$ all
  execute according to their specification.  For $\instrargs{k+j}$, if $j$ is
  odd, then it is a no-operation and therefore it changes no location values, and if $j$ is
  even, then it executes according to specification because it is executing in
  an initial state.  Let $\memloc$ be some original location whose value is
  incorrect after the buggy instruction $\instrargs{k} = \instrargs{b}$ executes (we know the buggy instruction corrupts an original
  location because it does not corrupt a duplicate location), with
  $\memlocdup = \funcmemlocscorr(\memloc)$.  We consider a type-B bug
  as in Definition~\ref{def:bug:specific:qed:test} first and assume
  that $\memloc$ is not the output location of $\instrargs{k}$. Since
  $\memstateargs{0}$ is
  QED-consistent, $\memstateargs{0}(\memloc) = \memstateargs{0}(\memlocdup)$.
  Since instructions $1$ through $k$ do not change duplicate locations, we have
  $\memstateargs{0}(\memloc) = \memstateargs{k}(\memlocdup)$.
  By repeated application of Lemma~\ref{def:dup:specvalue} and by definition of
  no-operation instructions, we can then
  conclude that $\memstateargs{k-1}(\memloc) =
  \memstateargs{3k-1}(\memlocdup)$.  Then, because of the bug, it
  follows that $\memstateargs{k}(\memloc) \not= \memstateargs{3k}(\memlocdup)$.  Finally, because
  none of the instructions after $\instrargs{k}$ modify original locations,
  $\memstateargs{k}(\memloc) = \memstateargs{3k}(\memloc)$, so $\neg
  \qedcons(\memstateargs{3k})$.

  The case of a type-A bug where $\memloc$ is the output location of
  $\instrargs{k}$ can be proved analogously. \qedhere
\end{proof}

\begin{proof}[Proof of Theorem~\ref{thm:hard:reset}]
  Suppose $\transsys$ is not correct.  Let $\instrseq^b = \langle
  \instrargs{1},\ldots,\instrargs{k} \rangle$ be a sequence required to reach a
  state $\memstateargs{b}\in \triggerstatesnoarg$ starting from some initial
  state $\meminitstate$, for some bug $\bug = \langle \instr_b,
  \triggerstatesnoarg \rangle$, such that the buggy instruction
  $\instr_b = \instrargs{k}$ is triggered in state $\memstateargs{b}$.
  We know
  $k \ge 2$ since $\transsys$ is single-instruction correct.
  Let $\instrseq$ be the unique bug-specific hard-reset QED test with initial
  state $\meminitstate$ whose prefix is $\instrseq^b$.
  Let $\statepath = \transrelmulti(\meminitstate, \instrseq)$.  It is easy to
  see that $\memstateargs{k-1} = \memstateargs{2k}$, because both are the result
  of executing the same instructions from the same initial state.  Furthermore,
  because $\instrargs{2k+1}$ is a no-operation instruction, all locations in
  $\memstateargs{2k+1}$ have the same values as those in $\memstateargs{2k}$.
  Let $\memloc$ be some location whose value is incorrect in $\memstateargs{k}$
  after $\instrargs{k}$ executes.  The same location must have the correct
  value in $\memstateargs{2k+2}$ because $\instrargs{2k+2}$ executes from the
  same architectural state and must execute correctly because it is an initial
  state and $\transsys$ is single-instruction correct.  Thus
  $\memstateargs{k}(\memloc) \not= \memstateargs{2k+2}(\memloc)$, and the
  hard-reset QED test fails, which is a contradiction.  Therefore, $\transsys$
  must be correct. Note that the above reasoning applies to both
  type-A and type-B bugs as in
  Definition~\ref{def:bug:specific:qed:test}. \qedhere
\end{proof}

\end{appendices}

\end{document}